\documentclass[11pt]{article}

\usepackage{amsmath,amssymb,amsfonts,amsthm}
\usepackage{graphicx}
\usepackage{authblk}
\usepackage{xcolor}
\usepackage{placeins}
\usepackage[margin=1in]{geometry}

\newtheorem{theorem}{Theorem}
\newtheorem{lemma}[theorem]{Lemma}
\newtheorem{corollary}[theorem]{Corollary}
\newtheorem{observation}[theorem]{Observation}

\newcommand{\N}{\mathbb{N}}
\newcommand{\Q}{\mathbb{Q}}
\newcommand{\R}{\mathbb{R}}
\newcommand{\Z}{\mathbb{Z}}

\DeclareMathOperator{\poly}{poly}
    
\begin{document}

\title{Lines in Every Direction with No ee-Random Points}

\author{Neil Lutz}
\author{Spencer Park Martin}
\author{Rain White}
\affil{Swarthmore College}

\date{}

\maketitle

\begin{abstract}
    We prove that in every direction in the Euclidean plane, there exists a line containing no double exponential time random (ee-random) points. This means each point on these lines has an algorithmically predictable location, to the extent that a gambler in an environment with fair payouts can, using double exponential time computing resources, amass unbounded capital placing bets on increasingly precise estimates of the point's location. Our proof relies on effectivizing the construction of the lineal extension of a Kakeya set. This resolves an open question of Lutz and Lutz (2015).
\end{abstract}

\section{Introduction}
    The theory of algorithmic information is fruitfully applied in continuous geometric spaces by regarding individual points as infinite data streams and quantifying the predictability or compressibility of those streams according to various notions of algorithmic randomness or dimension. There has been significant recent work using this type of analysis to study points on lines~\cite{LutLut15,LutLut18,LutStu20,LutStu22,Stul25,BusFie25}. In this paper, we show that in every direction in $\R^2$, there is a line on which no point is \emph{double exponential time random} (or ee-\emph{random}). We define this algorithmic randomness notion in Section~\ref{ssec:martingalesmeasure}, but informally it means that every point on the line is at least slightly predictable to an algorithm running in double exponential time. Thus a gambler with double exponential time computing power could, in an environment with fair payouts, accumulate unbounded capital by placing bets on the location of the point.
    
    This paper resolves an open question posed in 2015 by Lutz and Lutz~\cite{LutLut15}, which addressed two related questions by proving the existence of lines in all directions containing no computably random points and line segments in every direction containing no double exponential time random points. The authors asked whether the latter result could be extended to include full lines. Such extensions are often nontrivial in geometric measure theory. Keleti~\cite{Kele16} proved, for example, that given any union of line segments in $\R^2$, taking the \emph{lineal extension}---the union of all lines containing these segments---does not change the Hausdorff dimension, but it is open whether this holds in higher-dimensional Euclidean spaces. Both the results of~\cite{LutLut15} and our new result are obtained by effectivizing constructions of different \emph{Kakeya sets}.
        
    Algorithmic information theoretic analysis of points on lines is closely tied to Kakeya sets. In this paper, a Kakeya set is a set in a Euclidean space $\R^n$ that contains unit-length line segments in all directions. Besicovitch~\cite{Besi19,Besi28} proved in 1919 that such a set can have Lebesgue measure 0. Revisiting the problem decades later, Besicovitch~\cite{Besi64} in 1964 gave a simpler construction of a set with Lebesgue measure 0 that contains full lines in all directions, and Davies~\cite{Davi71} proved in 1971 that every Kakeya set in $\R^2$ has Hausdorff dimension 2. Since then, there has been a prolific line of research studying the properties of Kakeya sets and their lineal extensions, including a major 2025 breakthrough in which Wang and Zahl~\cite{WanZah25} proved that every Kakeya set in $\R^3$ has Hausdorff dimension 3. Our result can be equivalently stated as saying there is a Kakeya set in $\R^2$ whose lineal extension contains no double exponential time random points.

\section{Double Exponential Time Measure and Randomness in $\R^2$}
We now define double exponential time measure 0 and double exponential time randomness in $\R^2$ using martingales. These definitions generalize straightforwardly to other resource bounds and to higher-dimensional Euclidean spaces~\cite{LutLut15}.

\subsection{Dyadic Squares and Martingales}
For each $r\in\N$, the \emph{$r$-dyadic rationals} are the elements of the set $2^{-r}\Z$.

For each $r\in\N$ and $u,v\in\Z$, let
\[Q_r(u,v)=2^{-r}([u,u+1)\times[v,v+1))\]
be the (half-open, half-closed) \emph{$r$-dyadic square} at $(u,v)$, and let
\[\mathcal{Q}=\left\{Q_r(u,v):r\in\N\text{ and }u,v\in\{0,\ldots,2^r-1\}\right\}\]
be the set of all dyadic squares in $[0,1)^2$.

A \emph{martingale} on $[0,1)^2$ is a function $d:\mathcal{Q}\to[0,\infty)$ satisfying
\begin{equation}\label{eq:martingale}
    d(Q_r(u,v))=\frac{1}{4}\sum_{\alpha,\beta\in\{0,1\}}d(Q_{r+1}(2u+\alpha,2v+\beta)).
\end{equation}
for all $Q_r(u,v)\in\mathcal{Q}$. This condition places a limit on how quickly the martingale's value can grow with $r$; it implies, for each $r$-dyadic square $Q_r(u,v)$ and all $\alpha,\beta\in\{0,1\}$, that
\[d(Q_{r+1}(2u+\alpha,2v+\beta))\leq 4d(Q_r(u,v)),\]
so
\begin{equation}\label{eq:growthbound}
    d(Q_r(u,v))\leq 4^r d(Q_0(0,0)).
\end{equation}
holds for all $Q_r(u,v)\in\mathcal{Q}$. We call $d(Q_0(0,0))$ the \emph{initial capital} of the martingale $d$.

Informally, we regard a martingale as a betting strategy for a gambler placing bets on increasingly precise estimates for the location of an unknown point---which may be adversarially chosen from some known set---in an environment with fair payouts. When the point is revealed to be in the dyadic square $Q$, the gambler has amassed $d(Q)$ capital, which they now allocate among the four immediate dyadic subsquares of $Q$. The point will be revealed to be in one of those subsquares, at which point the gambler will receive capital equal to four times the amount they allocated to that subsquare. As there were four subsquares to choose from, this payout is fair. The amount allocated to the other subsquares is forfeited by the gambler, and the process is repeated ad infinitum on ever-smaller squares.

For each martingale $d$, $r\in\N$, and $x,y\in[0,1)$, define
\[d^{(r)}(x,y)=d(Q_r(\lfloor 2^rx\rfloor,\lfloor 2^ry\rfloor)),\]
the value of the martingale on the unique $r$-dyadic square containing the point $(x,y)$. The martingale $d$ \emph{succeeds} at a point $(x,y)\in[0,1)^2$ if
\[\limsup_{r\to\infty}d^{(r)}(x,y)=\infty,\]
and its \emph{success set}
\[S^\infty[d]=\left\{(x,y)\in[0,1)^2:\limsup_{r\to\infty}d^{(r)}(x,y)=\infty\right\}\]
is the set of points on which it succeeds.

\subsection{Martingales and Measure}\label{ssec:martingalesmeasure}
Originally working in the Cantor space, Ville~\cite{Vill39} proved that a set $E$ has Lebesgue measure $\lambda(E)=0$ if and only if it is contained in the success set of some martingale. This is also true for $E\subseteq[0,1)^2$, and this correspondence is extended to the entire plane by taking a union of unit square intersections with the set: Given any set $E\subseteq\R^2$, define
\begin{equation}\label{eq:hash}
    \#E=[0,1)^2\cap\bigcup_{i,j\in\Z}E-(i,j).
\end{equation}
Then $\lambda(E)=0$ if and only if there is a martingale $d$ such that $\#E\subseteq S^\infty[d]$.

Effective notions of measure 0 are obtained by placing resource bounds on this martingale. A martingale $d:\mathcal{Q}\to[0,\infty)$ is \emph{double exponential time computable} if there is a function
\[\hat{d}:\mathcal{Q}\times \N\to\Q\]
such that, for all $Q_r(u,v)\in\mathcal{Q}$ and all $t\in\N$,
\begin{equation}\label{eq:dhat}
    \big|d(Q_r(u,v))-\hat{d}(Q_r(u,v),t)\big|<2^{-t}
\end{equation}
and $\hat{d}(Q_r(u,v),t)$ is computable in time $2^{2^{O(r+t)}}$.

A set $E$ has \emph{double exponential time measure 0} (or ee\emph{-measure 0}) if there is some double exponential time computable martingale $d$ such that $\#E\subseteq S^\infty[d]$. This type of \emph{resource-bounded measure} was originally used to quantify the size of complexity classes~\cite{Lutz92} but is also useful for classifying geometric sets of Lebesgue measure 0. A point $(x,y)\in\R^2$ is \emph{double exponential time random} (or ee\emph{-random}) if it does not belong to any set of ee-measure 0.

\subsection{Open Set Martingales}\label{ssec:osm}

We will make essential use of a class of martingales defined in~\cite{LutLut15} which, for some given set $E$, for each dyadic square $Q$, allocate value to the four immediate dyadic subsquares of $Q$ proportional to the measure of each subsquare's intersection with $E$. Formally, for any set $E\subseteq[0,1)^2$, the martingale $d_E:\mathcal{Q}\to[0,\infty)$ is defined as follows.
\begin{itemize}
    \item The initial capital is
    \[
        d_E(Q_0(0,0))=
        \begin{cases}
           0&\lambda(E)=0\\
         1&\text{otherwise.}
        \end{cases}
    \]
    \item For all $r\in\N$, $u,v\in\{0,\ldots,2^r-1\}$ and $\alpha,\beta\in\{0,1\}$, the value
    \[d_E(Q_{r+1}(2u+\alpha,2v+\beta))\]
    is $0$ if $d_E(Q_r(u,v))=0$ and is otherwise given by
    \[4d_E(Q_r(u,v))\frac{\lambda(E\cap Q_{r+1}(2u+\alpha,2v+\beta))}{\lambda(E\cap Q_r(u,v))}.\]
\end{itemize}
When $E$ is an open set, $d_E$ is called the \emph{open set martingale} for $E$.

\begin{theorem}[Lutz and Lutz~\cite{LutLut15}]\label{thm:osm}
    For every nonempty set $E$ that is open as a subset of $[0,1)^2$ and every $(x,y)\in E$,
    \[d_E^{(r)}(x,y)=\frac{1}{\lambda(E)}\]
    for all sufficiently large $r$.
\end{theorem}

\section{Constructing the Set}

We follow a Kakeya set construction from the textbook of Bishop and Peres~\cite{BisPer17}, who attribute this construction to Bishop while also acknowledging an earlier, related construction by Sawyer~\cite{Sawy87}. The authors prove that the closure of the set has Lebesgue measure 0 and further that $\delta$-neighborhoods of the set have Lebesgue measure $O\left(\frac{1}{-\log\delta}\right)$.

Here we modify that construction---in a manner loosely suggested by an exercise in~\cite{BisPer17}---to construct a set that contains full lines in all directions. In Section~\ref{sec:mainthm} we will show that this set has ee-measure 0.

Define the sequence $\{x_k\}_{k\in\N}$ of dyadic rationals as the concatenation, over all $\ell\in\N$, of the blocks
\[X_\ell=\left((-1)^\ell(i2^{-\ell}-\ell):i=0,\ldots,2^\ell(2\ell+1)-1\right).\]
Thus for even $\ell$ the block $X_\ell$ consists of all $\ell$-dyadic rationals in the range $[-\ell,\ell+1)$ in ascending order, and for odd $\ell$ the block $X_\ell$ consists of all $\ell$-dyadic rationals in the range $(-\ell-1,\ell]$ in descending order.

Using this sequence, define the function $b:[0,1]\to\R$ by
\begin{equation}\label{eq:b}
    b(a)=\sum_{n=1}^\infty(x_{n-1}-x_n)(a-2^{-n}\lfloor a2^{n}\rfloor).
\end{equation}
For each slope $a\in [0,1]$, our set $F$ will contain the line with slope $a$ and $y$-intercept $b(a)$: We define
\[F=\{(x,ax+b(a)):a\in[0,1],x\in\R\}.\]

\section{Finitary Precursors to $F$}\label{sec:precursors}

We now define and analyze finitary precursors to the set $F$, which will be used to build an effective martingale that succeeds everywhere in $\#F$, the set obtained by applying the $\#$ operator defined in~\eqref{eq:hash} to $F$.

For each positive integer $k$, define the function $b_k:[0,1]\to\R$ by
\begin{equation}\label{eq:bk}
    b_k(a)=\sum_{n=1}^k(x_{n-1}-x_n)(a-2^{-n}\lfloor 2^{n}a\rfloor),
\end{equation}
and define the set
\[F_k= \{ (x, ax + b_k(a)) : a\in (2^{-k}\Z) \cap[0,1), x \in \R \}.\]
Thus $F_k$ is the union of $2^k$ lines, each of which has a slope and intercept that can be explicitly calculated in finite time.

For every set $E\subseteq\R^2$ and every $\delta>0$, let $E[\delta]$ denote the open $\delta$-neighborhood of $E$ under the $\ell^\infty$ norm. That is, a point $(x,y)\in\R^2$ belongs to $E[\delta]$ if and only if there is some point $(x',y')\in E$ such that
\[\max\{|x-x'|,|y-y'|\}<\delta.\]
Using the $\ell^\infty$ norm instead of the ordinary Euclidean norm will slightly simplify our analysis, in part because neighborhoods of polygons under the $\ell^\infty$ norm are polygons.

For every $m\in\N$ and every $m$-dyadic rational $\hat{x}\in [-m,m]$, define the set
\[\tilde{F}_{m,\hat{x}}=(F_k\cap([\hat{x},\hat{x}+2^{-m})\times\R))[2^{1-m-k}],\]
where $k$ is the unique index such that $\hat{x}=x_k$ is in the block $X_m$. For our eventual double exponential time bound, it is important that this $k$ is only exponentially larger than $m$.

\begin{observation}\label{obs:kbound}
    For every integer $m\geq 4$, every $x_k$ in the block $X_m$ satisfies $k\leq 2^{2m}$.
\end{observation}
\begin{proof}
    Every $x_k$ in $X_m$ satisfies
    \begin{align*}
        k&\leq\sum_{\ell=1}^{m}|X_\ell|\\
        &=\sum_{\ell=1}^{m} \big(2^\ell(2\ell+1)-1\big)\\
        &=(2^{m+2}-1)m-2(2^m-1)\\
        &\leq 2^{2m},
    \end{align*}
    where the last inequality holds whenever $m\geq 4$.
\end{proof}

The set $\tilde{F}_{m,\hat{x}}$ can be obtained by taking the union of $2^{k}$ thin tubes, i.e., neighborhoods of lines, and then taking a thin vertical section of that union. As we now show, this vertical section must have small measure. Furthermore, within the scope of an even thinner vertical section, these $2^k$ tubes cover the union of the infinitely many lines in $F$.

\begin{lemma}\label{lem:maintech}
    For every $m\in\N$ and every $m$-dyadic rational $\hat{x}\in[-m,m]$, we have
    \begin{equation}\label{eq:fmxarea}
        \lambda(\tilde{F}_{m,\hat{x}})\leq 2^{2-2m}
    \end{equation}
    and
    \begin{equation}\label{eq:fmxcontainment}
        F\cap([\hat{x},\hat{x}+2^{-m})\times\R)\subseteq \tilde{F}_{m,\hat{x}}\cap([\hat{x},\hat{x}+2^{-m})\times\R).
    \end{equation}    
\end{lemma}
\begin{proof}
    Fix any $m\in\N$ and $\hat{x}\in(2^{-m}\Z)\cap[-m,m]$. Letting
    \[\delta=2^{1-m-k},\]
    observe that $\tilde{F}_{m,\hat{x}}$ can equivalently be written as
    \[\bigcup_{a\in(2^{-k}\Z)\cap[0,1)}L_{a,b_k(a)}[\delta]\cap ([\hat{x},\hat{x}+2^{-m})\times\R)[\delta],\]
    where $L_{a,b_k(a)}$ denotes the line of slope $a$ and $y$-intercept $b_k(a)$. Each of these lines intersects
    \[([\hat{x},\hat{x}+2^{-m})\times\R)[\delta]\]
    in a segment of horizontal length $2^{-m}+2\delta$, and each vertical slice of $L_{a,b_k(a)}[\delta]$ has length $2\delta$. Therefore the area of the intersection for each $L_{a,b_k(a)}$ is at most $(2^{-m}+2\delta)\cdot 2\delta$. As there are $2^{k}$ lines, the area of the union of these intersections is at most
    \[2^k\cdot(2^{-m}+2\delta)\cdot 2\delta\leq 2^{2-2m},\]
    proving that~\eqref{eq:fmxarea} holds.
    
    For~\eqref{eq:fmxcontainment}, observe first that the function $b_k$ defined in~\eqref{eq:bk} closely approximates the function $b$ defined in~\eqref{eq:b}, in the sense that for all $a\in[0,1]$ we have
    \begin{align*}
        |b(a)-b_k(a)|&\leq\sum_{n=k+1}^\infty |x_{n-1}-x_n|(a-2^{-n}\lfloor 2^na\rfloor)\\
        &\leq 2^{-m}\sum_{n=k+1}^\infty 2^{-n}\\
        &=2^{-m-k}.
    \end{align*}

    The function $b_k(a)$ is piecewise linear and right-continuous, with slope $-x_k=-\hat{x}$ at all $a\in[0,1]\setminus 2^{-k}\Z$. Therefore for all $x\in[\hat{x},\hat{x}+2^{-m})$, $ax+b_k(a)$ is also piecewise linear and right-continuous as a function of $a$, with slope of absolute value
    \begin{align*}
        \left|\frac{d}{da}(ax+b_k(a))\right|&=|x-\hat{x}|\\
        &<2^{-m}
    \end{align*}
    at all $a\in[0,1]\setminus 2^{-k}\Z$.

    Let $a\in[0,1]$ and $x\in[\hat{x},\hat{x}+2^{-m})$ be such that
    \[(x,ax+b(a))\in F,\]
    and let $\hat{a}=2^{-k}\lfloor a2^k\rfloor$. Then
    \[(x,\hat{a}x+b_k(\hat{a}))\in F_k,\]
    and
    \begin{align*}
        |ax+b(a)-(\hat{a}x+b_k(\hat{a}))|
        &\leq |b(a)-b_k(a)|+|ax+b_k(a)-(\hat{a}x+b_k(\hat{a}))|\\
        &\leq 2^{-m-k}+2^{-m}|a-\hat{a}|\\
        &< 2^{1-m-k}.
    \end{align*}
    The vertical slice at $x$ of $\tilde{F}_{m,\hat{x}}$ includes a $2^{1-m-k}$ open neighborhood of $\hat{a}x+b_k(\hat{a})$. This shows that
    \[(x,ax+b(a))\in \tilde{F}_{m,\hat{x}}\cap([\hat{x},\hat{x}+2^{-m})\times \R).\]
    As $x\in[\hat{x},\hat{x}+2^{-m})$ was arbitrary, this completes the proof of~\eqref{eq:fmxcontainment} and of the lemma.
\end{proof}

\section{Martingales for Finitary Precursor Sets}

\subsection{Open Set Martingales for the $\Theta(2^{-m})\times 1$ Scale}

Lemma~\ref{lem:maintech} tells us that $F$ is covered by a collection of thin vertical sets, each of which has small measure. We now further divide these vertical sets into pieces of height 1, translate those pieces to the unit square $[0,1)^2$, and define an effective martingale corresponding to each such piece.

For each $m\in\N$, $\hat{x}\in (2^{-m}\Z)\cap[-m,m]$, and $j\in\Z$, let
\[\tilde{F}_{m,\hat{x},j}=[0,1)^2\cap(\tilde{F}_{m,\hat{x}}-(\lfloor\hat{x}\rfloor,j)),\]
and let $\tilde{d}_{m,\hat{x},j}$ be the open set martingale (as defined in Section~\ref{ssec:osm}) for $\tilde{F}_{m,\hat{x},j}$.

\begin{lemma}\label{lem:fmxj}
    For every $m\in\N$, $\hat{x}\in 2^{-m}\Z\cap[-m,m]$, $j\in\Z$, and $Q_r(u,v)\in\mathcal{Q}$, the martingale value
    \[\tilde{d}_{m,\hat{x},j}(Q_r(u,v))\]
    is computable in $2^{2^{O(m+r)}}$ time.
\end{lemma}
\begin{proof}
    Fix any integer $m\geq 4$, $\hat{x}\in 2^{-m}\Z\cap[-m,m]$, and $j\in\Z$. We will prove the lemma by induction on $r$. It is trivially true for $r=0$, so fix $r\geq 0$ and assume that $\tilde{d}_{m,\hat{x},j}(Q)$ is computable in $2^{2^{O(m+r)}}$ time for each $r$-dyadic square $Q\in\mathcal{Q}$.
    
    Let $Q_{r+1}(2u+\alpha,2v+\beta)$, where $u,v\in\{0,\ldots,2^r-1\}$ and $\alpha,\beta\in\{0,1\}$, be any $(r+1)$-dyadic square. Then $\tilde{d}_{m,\hat{x},j}(Q_{r+1}(2u+\alpha,2v+\beta))$ is given by either 0 or
    \[4\tilde{d}_{m,\hat{x},j}(Q_r(u,v))\frac{\lambda(\tilde{F}_{m,\hat{x},j}\cap Q_{r+1}(2u+\alpha,2v+\beta))}{\lambda(\tilde{F}_{m,\hat{x},j}\cap Q_r(u,v))},\]
    so it suffices to show that $\lambda(\tilde{F}_{m,\hat{x},j}\cap Q_{r+1}(2u+\alpha,2v+\beta))$ and $\lambda(\tilde{F}_{m,\hat{x},j}\cap Q_r(u,v))$ can be found in $2^{2^{O(m+r)}}$ time.
    
    Let $k$ be the unique index such that $\hat{x}=x_k$ is in the block $X_m$, and let $\delta=2^{1-m-k}$, noting that $k\leq 2^{2m}$ for all $m\geq 4$ by Observation~\ref{obs:kbound}. Then $\tilde{F}_{m,\hat{x},j}$ is the intersection with $[0,1)^2$ of a translated copy of $\tilde{F}_{m,\hat{x}}$, i.e., of
    \begin{equation}\label{eq:fmxalt}
        \bigcup_{a\in(2^{-k}\Z)\cap[0,1)}L_{a,b_k(a)}[\delta]\cap ([\hat{x},\hat{x}+2^{-m})\times\R)[\delta],
    \end{equation}
    where $L_{a,b(a)}$ denotes the line with slope $a$ and $y$-intercept $b(a)$. We can find the set~\eqref{eq:fmxalt} as follows. Initially set $S=\emptyset$. For $a\in(2^{-k}\Z)\cap[0,1)$, we do the following.
    
    First find $b_k(a)$ by directly applying~\eqref{eq:bk}, which involves finding the sum of $2^k$ terms, each of which is an $m+k\leq (2k)$-dyadic number. This runs in $\poly(2^k)$ time.
    
    Then set
    \[S\leftarrow S\cup \big(L_{a,b_k(a)}[\delta]\cap ([\hat{x},\hat{x}+2^{-m})\times\R)[\delta]\big).\]
    Each $L_{a,b_k(a)}[\delta]\cap ([\hat{x},\hat{x}+2^{-m})\times\R)[\delta]$ is a parallelogram, and at all stages, each vertex of $S$ is the intersection of segments on the boundaries of these $2^k$ parallelograms. Therefore $S$ never has more than $\binom{4\cdot 2^k}{2}<2^{4+2k}$ vertices. Standard algorithms from computational geometry can find the union of two polygons, each with at most $N$ vertices, in $\poly(N)$ time~\cite{dBCKO08}, so each union operation takes $\poly(2^{4k+2})=\poly(2^k)$ time. In total, finding the polygon $\tilde{F}_{m,\hat{x}}$ takes $2^k(\poly(2^k)+\poly(2^k))=\poly(2^k)$ time.

    Translating $\tilde{F}_{m,\hat{x}}$ and finding the area of the translated polygon's intersection with $Q_r(u,v)$ takes $\poly(2^k+r)$ time, again using standard algorithms from computational geometry~\cite{dBCKO08}. Therefore the denominator $\lambda(\tilde{F}_{m,\hat{x},j}\cap Q_r(u,v))$ can be found in $\poly(2^k+r)\leq 2^{2^{O(m+r)}}$ time. For the same reason, the numerator $\lambda(\tilde{F}_{m,\hat{x},j}\cap Q_{r+1}(2u+\alpha,2v+\beta))$ can be found in $2^{2^{O(m+r)}}$ time, completing the proof of the lemma.
\end{proof}

\subsection{Combining Martingales for the $1\times 1$ Scale}\label{ssec:combining}
We now average the $2^m$ martingales $\tilde{d}_{m,\hat{x},j}$ corresponding to each integer unit square piece of $F$ to get an effective martingale corresponding to that whole $1\times 1$ square.

For every $m\in\N$ and $i,j\in\Z$, let $d_{m,i,j}:\mathcal{Q}\to[0,\infty)$ be the martingale defined by
\begin{equation}\label{eq:dmij}
    d_{m,i,j}(Q)=2^{-m}\sum_{\hat{x}\in (2^{-m}\Z)\cap[i,i+1)}\tilde{d}_{m,\hat{x},j}(Q).
\end{equation}
By Lemma~\ref{lem:fmxj}, $d_{m,i,j}(Q)$ is computable in
\[2^m\cdot 2^{2^{O(m+r)}}\leq 2^{2^{O(m+r)}}\]
time for each $r$-dyadic square $Q\in\mathcal{Q}$.

\begin{lemma}\label{lem:dmijr}
    Let $i,j\in\Z$ and $m\geq\max\{|i|+1,4\}$. For every $(x,y)\in F\cap([i,i+1)\times[j,j+1))$,
    for all sufficiently large $r\in\N$,
    \[d_{m,i,j}^{(r)}(x-i,y-j)\geq 2^{m-2}.\]
\end{lemma}
\begin{proof}
    For each $\hat{x}\in (2^{-m}\Z)\cap[i,i+1)$,
    \[\tilde{F}_{m,\hat{x},j}\subseteq\tilde{F}_{m,\hat{x}}-(i,j).\]
    As
    \begin{align*}
        \lambda(\tilde{F}_{m,\hat{x}}-(i,j))&=\lambda(\tilde{F}_{m,\hat{x}})\\
        &\leq 2^{2-2m}
    \end{align*}
    by Lemma~\ref{lem:maintech}, Theorem~\ref{thm:osm} tells us that for all sufficiently large $r$ and all $(x,y)\in F\cap([\hat{x},\hat{x}+2^{-m})\times [j,j+1))$,
    \[\tilde{d}_{m,\hat{x},j}^{(r)}(x-i,y-j)\geq 2^{2m-2}.\]
    Therefore for all sufficiently large $r$ and all $(x,y)\in F\cap([i,i+1)\times[j,j+1)$, letting $\hat{x}=2^{-m}\lfloor 2^mx\rfloor$, we have
    \begin{align*}
        d_{m,i,j}^{(r)}(x-i,y-j)&\geq 2^{-m}\tilde{d}_{m,\hat{x},j}^{(r)}(x-i,y-j)\\
        &\geq 2^{-m}\cdot 2^{2m-2}\\
        &=2^{m-2},
    \end{align*}
    which is the desired bound.
\end{proof}

\section{A Martingale for $\#F$}

In Section~\ref{ssec:combining}, we constructed, for each $m\in\N$ and $i,j\in\Z$, a martingale that attains value $\Omega(2^m)$ on the piece of $F$ that lies in the unit square with lower left-hand corner $(i,j)$ (after it has been translated). To attain unbounded value on all of $F$ (again, after being translated), we now sum these martingales over all $m$, $i$, and $j$, with exponential discounting to ensure that the resulting sum is a martingale with finite initial capital. 

Let $\pi:\Z^2\to\N$ be any standard pairing function, and define the function $\hat{d}:\mathcal{Q}\times\N\to\Q$ by
\[\hat{d}(Q_r(u,v),t)=\sum_{m=0}^{2r+t+1}\sum_{\substack{i,j\in\Z\\\pi(i,j)< m}} 2^{-m-\pi(i,j)}d_{m,i,j}(Q_r(u,v))\]
and the function $d:\mathcal{Q}\to [0,\infty)$ by
\[d(Q_r(u,v))=\sum_{m\in\N}\sum_{\substack{i,j\in\Z\\\pi(i,j)< m}} 2^{-m-\pi(i,j)}d_{m,i,j}(Q_r(u,v)).\]

\begin{lemma}\label{lem:eemartingale}
    The function $d$ is a double exponential time computable martingale.
\end{lemma}
\begin{proof}
    The martingale property~\eqref{eq:martingale} is preserved by linear combinations, so to show that $d$ is a martingale we only need to verify that its initial capital $d(Q_0(0,0))$ is finite.
    
    Each $\tilde{d}_{m,\hat{x},j}$ is an open set martingale, so each one's initial capital $\tilde{d}_{m,\hat{x},j}(Q_0(0,0))$ is either 0 or 1. It is immediate from~\eqref{eq:dmij} that each $d_{m,i,j}(Q_0(0,0))$ is therefore at most 1. As $\pi$ is a pairing function, we have
    \begin{align*}
        \sum_{\substack{i,j\in\Z\\\pi(i,j)< m}} 2^{-\pi(i,j)}d_{m,i,j}(Q_0(0,0))
        &\leq \sum_{\substack{i,j\in\Z\\\pi(i,j)< m}} 2^{-\pi(i,j)}\\
        &\leq\sum_{i,j\in\Z} 2^{-\pi(i,j)}\\
        &=\sum_{p\in\N} 2^{-p}\\
        &=2
    \end{align*}
    for all $m$, so
    \begin{align*}
        d(Q_0(0,0))&\leq\sum_{m\in\N}2^{1-m}\\
        &=4.
    \end{align*}
    
    To prove that $d$ is double exponential time computable, it suffices to show, for all $Q_r(u,v)\in\mathcal{Q}$ and all $t\in\N$, that 
    \[\big|d(Q_r(u,v))-\hat{d}(Q_r(u,v))\big|<2^{-t}\]
    and $\hat{d}(Q_r(u,v),t)$ is computable in $2^{2^{O(r+t)}}$ time.

    For the former condition, note that $d_{m,i,j}(Q_r(u,v))\leq 4^r$ by~\eqref{eq:growthbound}. Therefore
    \begin{align*}
        \big|d(Q_r(u,v))-\hat{d}(Q_r(u,v))\big|
        &\leq 4^r\sum_{m=2r+t+2}^\infty2^{-m}\sum_{\substack{i,j\in\Z\\\pi(i,j)< m}} 2^{-\pi(i,j)}\\
        &=2^{2r+1}\sum_{m=2r+t+2}^\infty2^{-m}\\
        &=2^{2r+1}2^{-2r-t-1}\\
        &=2^{-t}.
    \end{align*}
    
    To see that $\hat{d}(Q_r(u,v),t)$ is computable in $2^{2^{O(r+t)}}$ time, observe that computing each of the inner sums in $\hat{d}(Q_r(u,v),t)$ involves computing the martingale value $d_{m,i,j}(Q_r(u,v))$ for exactly $m$ pairs $(i,j)\in\Z^2$. As noted in Section~\ref{ssec:combining}, it follows from Lemma~\ref{lem:fmxj} that each of these martingale values is computable in $2^{2^{O(r+t)}}$ time, so the time to compute $m$ of them, scale each one by $2^{-\pi(i,j)}$, sum the results, and repeat this for $2r+t+2$ values of $m$, is still $2^{2^{O(r+t)}}$.
\end{proof}

\section{Main Theorem}\label{sec:mainthm}
\begin{theorem}\label{thm:main}
    There is a plane set of ee-measure 0 that contains lines in all directions.
\end{theorem}
\begin{proof}
    As the union of $F$ with three rotated copies of $F$ contains lines in all directions, it suffices to prove that $F$ has ee-measure 0. By Lemma~\ref{lem:eemartingale} and the definition of ee-measure 0 in Section~\ref{ssec:martingalesmeasure}, it suffices to show that $\#F\subseteq S^\infty[d]$.

    To this end, let $(x',y')\in\#F$. Then there are some $i,j\in\Z$ such that
    \[(x,y):=(x'+i,y'+j)\in F\cap([i,i+1)\times [j,j+1)).\]
    By Lemma~\ref{lem:dmijr}, for all integers $m\geq|i|+1$, for all sufficiently large $r\in\N$, we have
    \[d_{m,i,j}^{(r)}(x',y')\geq 2^{m-2}.\]
    Therefore
    \begin{align*}
        \limsup_{r\to\infty} d^{(r)}(x',y')&\geq\sum_{m=|i|+1}^\infty 2^{-m-\pi(i,j)}2^{m-2}\\
        &=\sum_{m=|i|+1}^\infty 2^{-\pi(i,j)-2}\\
        &=\infty.
    \end{align*}
    As the double exponential time martingale $d$ succeeds on every $(x',y')\in \#F$, we conclude that $F$ has ee-measure 0.
\end{proof}

By the definition of ee-randomness, the following corollary follows immediately.

\begin{corollary}\label{cor:main}
    For every direction $\theta\in[0,2\pi)$, there is a line in $\R^2$ with direction $\theta$ containing no ee-random points.
\end{corollary}

\section{Conclusion}
While we restricted our analysis to $\R^2$, higher-dimensional versions of Theorem~\ref{thm:main} and Corollary~\ref{cor:main} also hold. The union of $F\times\R^{n-2}$ with three rotated copies of itself contains lines in all directions in $\R^n$, and using the definitions and techniques of~\cite{LutLut15}, it is routine to show that this product also has ee-measure 0.

It is natural to ask if it is possible to improve on our double exponential time resource bound. Indeed,~\cite{LutLut15} conjectured there are lines in every direction that miss every polynomial time random point. Our effectivization, like those in~\cite{LutLut15}, relies on open set martingales and is fundamentally brute-force in character; our martingale $d$ is estimated by a function $\hat{d}(Q,t)$ that exhaustively iterates through a very large number of polygons. Known lower bounds on the measure of neighborhoods of Kakeya sets \cite{Cord93,Keic99,BisPer17} indicate that any Kakeya set with exponentially small area has a double exponential number of ``pieces.'' This presents a barrier to satisfying~\eqref{eq:dhat} with fewer than double exponentially many pieces, suggesting that improvements on the resource bound---such as showing that a Kakeya set's lineal extension can miss every exponential time random point---would require a different approach that implicitly or explicitly handles many pieces of a Kakeya set in aggregate.

\begin{figure}
\centering
\includegraphics[width=\columnwidth]{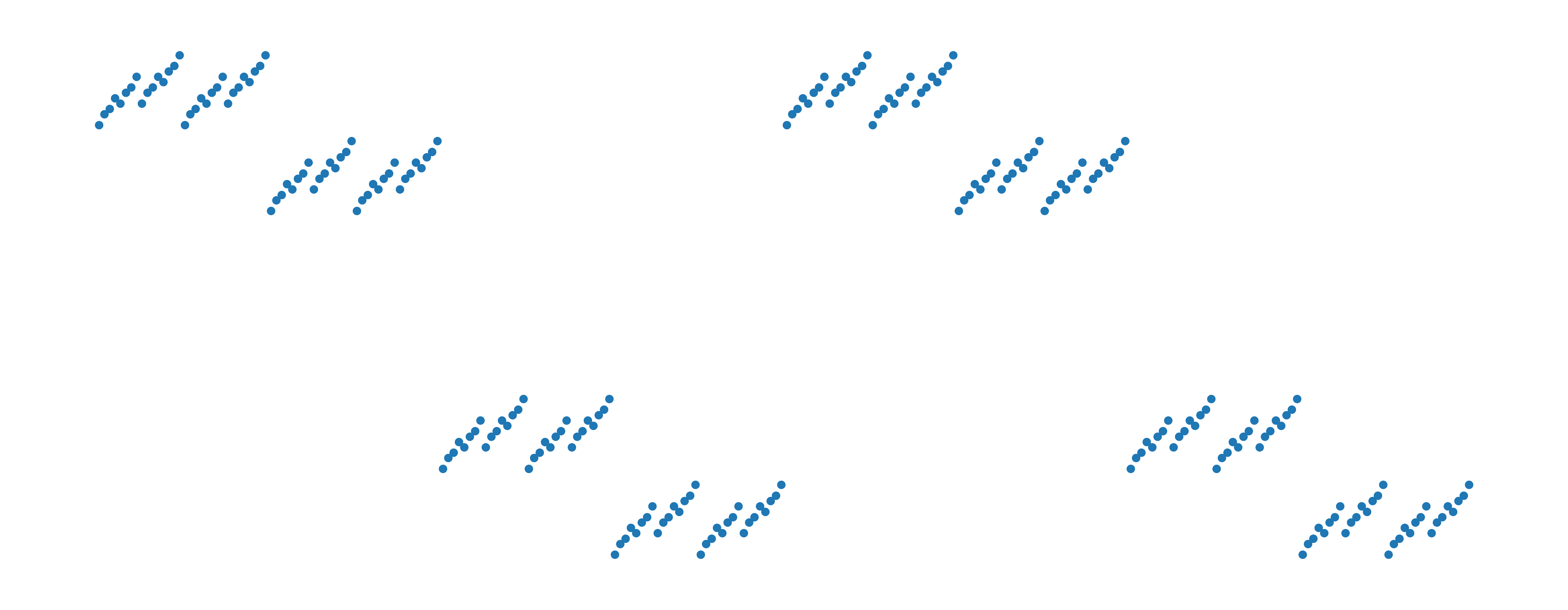}
\caption{The set $\{(a,b_8(a)):a\in(2^{-8}\Z)\cap [0,1)\}.$}
\end{figure}

\begin{figure}
\centering
\includegraphics[width=0.6\columnwidth]{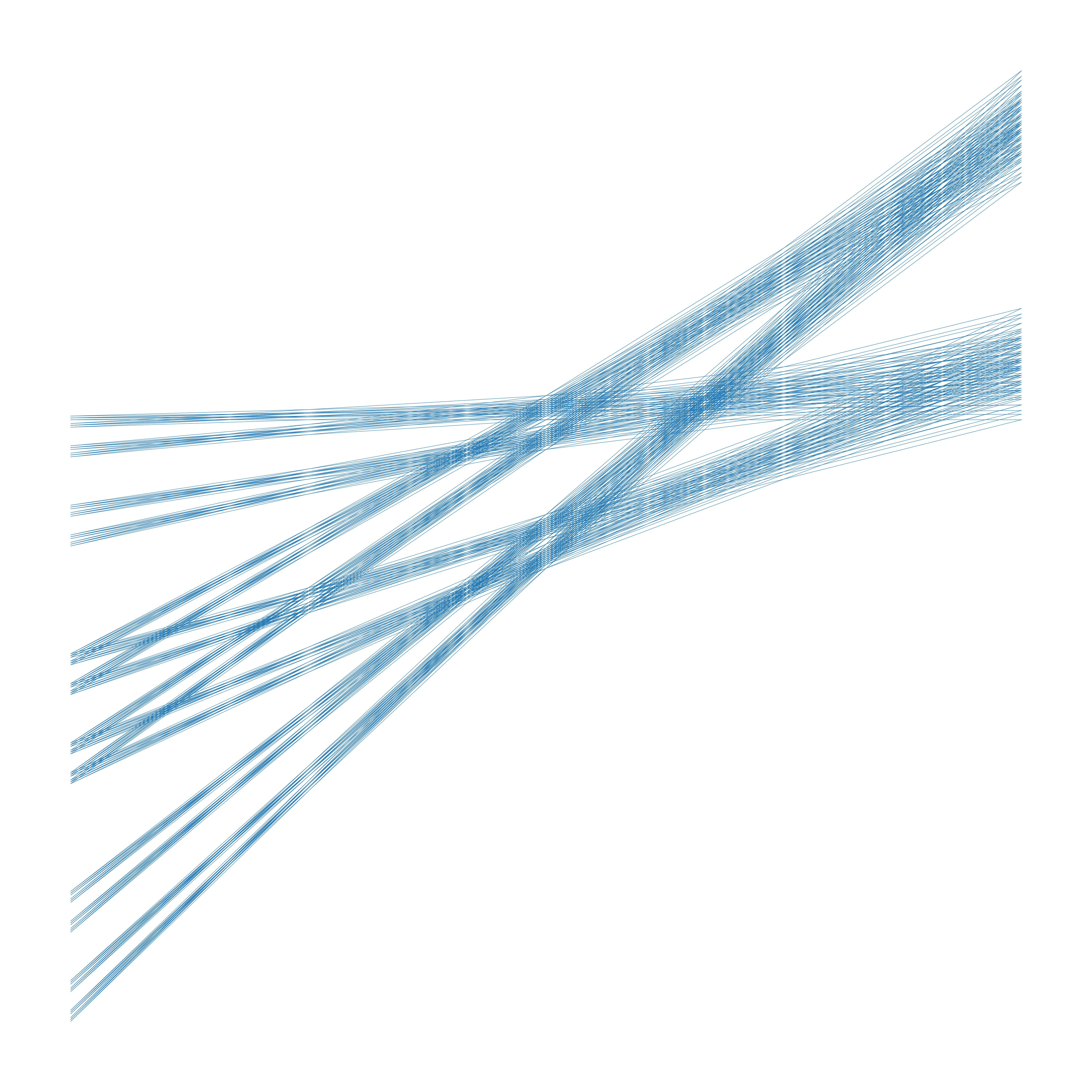}
\caption{The set $F_{8}\cap ([-1,1]\times\R)$.}
\end{figure}

\newpage
\bibliographystyle{abbrv}
\bibliography{leksmerp}

\begin{thebibliography}{10}

\bibitem{Besi19}
A.~S. Besicovitch.
\newblock Sur deux questions d'int{\'e}grabilit{\'e} des fonctions.
\newblock {\em Journal de la Soci{\'e}t{\'e} de physique et de math{\'e}matique de l{'}Universit{\'e} de Perm}, 2:105--123, 1919.

\bibitem{Besi28}
A.~S. Besicovitch.
\newblock On {K}akeya's problem and a similar one.
\newblock {\em Mathematische Zeitschrift}, 27:312--320, 1928.

\bibitem{Besi64}
A.~S. Besicovitch.
\newblock On fundamental geometric properties of plane line sets.
\newblock {\em Journal of the London Mathematical Society}, 39:441--448, 1964.

\bibitem{BisPer17}
C.~J. Bishop and Y.~Peres.
\newblock {\em Fractals in Probability and Analysis}.
\newblock Cambridge University Press, 2017.

\bibitem{BusFie25}
R.~E.~G. Bushling and J.~B. Fiedler.
\newblock Bounds on the dimension of lineal extensions.
\newblock {\em Journal of Fractal Geometry}, 12(1/2):105--133, 2025.

\bibitem{Cord93}
A.~Cordóba.
\newblock The fat needle problem.
\newblock {\em Bulletin of the London Mathematical Society}, 25(1):81--82, 01 1993.

\bibitem{Davi71}
R.~O. Davies.
\newblock Some remarks on the {K}akeya problem.
\newblock {\em Mathematical Proceedings of the Cambridge Philosophical Society}, 69:417--421, 1971.

\bibitem{dBCKO08}
M.~de~Berg, O.~Cheong, M.~van Kreveld, and M.~Overmars.
\newblock {\em Computational Geometry: Algorithms and Applications}.
\newblock Springer-Verlag, Santa Clara, CA, USA, 3rd edition, 2008.

\bibitem{Keic99}
U.~Keich.
\newblock On ${L}^p$ bounds for {K}akeya maximal functions and the {M}inkowski dimension in $\mathbb{R}^2$.
\newblock {\em Bulletin of the London Mathematical Society}, 31(2):213--221, 1999.

\bibitem{Kele16}
T.~Keleti.
\newblock Are lines much bigger than line segments?
\newblock {\em Proc. Amer. Math. Soc.}, 144:1535--1541, 2016.

\bibitem{Lutz92}
J.~H. Lutz.
\newblock Almost everywhere high nonuniform complexity.
\newblock {\em Journal of Computer and System Sciences}, 44(2):220--258, 1992.

\bibitem{LutLut15}
J.~H. Lutz and N.~Lutz.
\newblock Lines missing every random point.
\newblock {\em Computability}, 4(2):85--102, 2015.

\bibitem{LutLut18}
J.~H. Lutz and N.~Lutz.
\newblock Algorithmic information, plane {K}akeya sets, and conditional dimension.
\newblock {\em ACM Transactions on Computation Theory}, 10, 2018.

\bibitem{LutStu20}
N.~Lutz and D.~Stull.
\newblock Bounding the dimension of points on a line.
\newblock {\em Information and Computation}, 275, 2020.

\bibitem{LutStu22}
N.~Lutz and D.~Stull.
\newblock Dimension spectra of lines.
\newblock {\em Computability}, 11(2):85--112, 2022.

\bibitem{Sawy87}
E.~Sawyer.
\newblock Families of plane curves having translates in a set of measure zero.
\newblock {\em Mathematika}, 34(1):69--76, 1987.

\bibitem{Stul25}
D.~M. Stull.
\newblock The dimension spectrum conjecture for planar lines.
\newblock {\em Journal of the London Mathematical Society}, 111(6):e70216, 2025.

\bibitem{Vill39}
J.~Ville.
\newblock {\em \'{E}tude Critique de la Notion de Collectif}.
\newblock Gauthier--Villars, Paris, 1939.

\bibitem{WanZah25}
H.~Wang and J.~Zahl.
\newblock Volume estimates for unions of convex sets, and the {K}akeya set conjecture in three dimensions.
\newblock {\em arXiv preprint arXiv:2502.17655}, Feb 2025.

\end{thebibliography}

\end{document}